\UseRawInputEncoding
\pdfoutput=1
\documentclass[10pt,journal]{IEEEtran}
\usepackage{amsmath}
\usepackage[pdftex]{graphicx}
\usepackage{amsthm}
\usepackage{amsfonts}
\usepackage{epstopdf}
\usepackage{cite}
\usepackage{amssymb}
\usepackage{subfigure}
\usepackage{mathabx}
\usepackage{color}
\usepackage{amsmath,amssymb,latexsym,mathrsfs,bm,cite,color}
\usepackage{psfrag,epsfig,graphicx,subfigure,graphics,subfigure}
\usepackage{algorithmic,algorithm}
\usepackage{multirow}
\usepackage{threeparttable}
\usepackage{url}
\usepackage{booktabs}
\usepackage{bbm}
\usepackage[caption=false,font=normalsize,labelfont=sf,textfont=sf]{subfig}
\newtheorem{MyTheo}{Theorem}

\newtheorem{MyCoro}{Corollary}

\begin{document}
\title{Modelling and Performance Analysis of the Over-the-Air Computing in Cellular IoT Networks}
\author{\IEEEauthorblockN{
Ying Dong,
Haonan Hu, \textit{IEEE Member},
Qiaoshou Liu, Tingwei Lv, Qianbin Chen and Jie Zhang, \textit{IEEE Senior Member}}\\

\thanks{This work was supported in part by the National Natural Science Foundation of China (NSFC) under Grant 61831002; in part by the Science and Technology Research Program of Chongqing Municipal Education Commission under Grant KJZD-K202200604. (\textit{Corresponding author: Haonan Hu}.)

Ying Dong, Haonan Hu, Qiaoshou Liu and Qianbin Chen are with the School of Communication and Information Engineering, Chongqing University of Posts and Telecommunications, Chongqing 400 065, China (e-mail:  yingd.cqupt@qq.com; huhn@cqupt.edu.cn; liuqs@cqupt.edu.cn; chenqb@cqupt.edu.cn).

Tingwei Lv is with the International College, Chongqing University of Posts and Telecommunications, Chongqing 400 065, China (e-mail: 2020215128@cqupt.edu.cn). 

Jie Zhang is with the Department of Electronic and Electrical Engineering, University of Sheffield, S1 3JD Sheffield, U.K.(e-mail: jie.zhang@sheffield.ac.uk)}
}

\maketitle

\begin{abstract}
Ultra-fast wireless data aggregation (WDA) of distributed data has emerged as a critical design challenge in the ultra-densely deployed cellular internet of things network (CITN) due to limited spectral resources. Over-the-air computing (AirComp) has been proposed as an effective solution for ultra-fast WDA by exploiting the superposition property of wireless channels. However, the effect of access radius of access point (AP) on the AirComp performance has not been investigated yet. Therefore, in this work, the mean square error (MSE) performance of AirComp in the ultra-densely deployed CITN is analyzed with the AP access radius. By modelling the spatial locations of internet of things devices as a Poisson point process, the expression of MSE is derived in an analytical form, which is validated by Monte Carlo simulations. Based on the analytical MSE, we investigate the effect of AP access radius on the MSE of AirComp numerically. The results show that there exists an optimal AP access radius for AirComp, which can decrease the MSE by up to $12.7\%$. It indicates that the AP access radius should be carefully chosen to improve the AirComp performance in the ultra-densely deployed CITN.
\end{abstract}


\section{INTRODUCTION}
With the rapid growth of internet of things (IoT) applications, i.e., agricultural ecological environment monitoring and smart city, ultra-fast wireless data aggregation (WDA) of distributed data plays a key role in the future cellular IoT network (CITN) for data analysis, edge learning, and action determination \cite{9481289, 10038617}. However, the conventional orthogonal frequency division multiple access (OFDMA) scheme, i.e., all IoT devices transmit data to the access point (AP) via orthogonal wireless channels, may cause excessive latency because of limited spectral resources in the ultra-densely deployed CITN. Recently, over-the-air computing (AirComp) has been considered as a promising technique to achieve ultra-fast WDA in the ultra-densely deployed CITN \cite{9095231}. The AirComp utilizes the signal superposition property of wireless channels, which is able to directly process distributed sensor data by nomographic functions, e.g., arithmetic mean, weighted sum, and polynomial, over the air \cite{6557530,10092857}.

AirComp was first proposed in \cite{4305404}, where a computation code was developed to exploit channel collisions resulting from simultaneous transmission for reliable distributed computation over a multiple-access channel. Based on this work, in \cite{4655452}, the mean square error (MSE) was studied as the key metric for evaluating the AirComp performance. Subsequent research efforts have aimed to optimize the AirComp design to minimize the MSE \cite{8708985,9095231,9097897,9158563,9298474}. In \cite{8708985}, a wireless power-AirComp framework was proposed by jointly optimizing transmit power control and transmit energy to minimize the MSE. In \cite{9095231}, the MSE was minimized by optimizing the transmitting-receiving (Tx-Rx) policy of a single-antenna AirComp system. Based on this work, \cite{9097897} proposed an optimal Tx-Rx parameter design problem to minimize the MSE of AirComp under the constraint of transmit power of the IoT devices. Additionally, in \cite{9158563}, it was shown that the proposed optimal transmit power control policy can reduce the MSE of AirComp. Based on the result, in \cite{9298474}, an algorithm for controlling transmit power was proposed. The simulation results shown that the proposed algorithm can significantly decrease the MSE. However, these works mainly focus on the simulation-based analysis of AirComp performance, while lacking the analytical expression of MSE to evaluate the AirComp performance. Additionally, \cite{9624970} has investigated the AirComp performance in networks with randomly uniform distribution of devices and derived an approximate result of MSE. However, their study assumed that all devices in the random network perform AirComp without the AP access radius. Therefore, the above work cannot analyze the effect of AP access radius on the AirComp performance.

In this work, we analyze the AirComp performance in the CITN, where the locations of IoT devices are modeled following a Poisson point process (PPP). It is assumed that all IoT devices adopt the channel inversion power control mechanism under the constraint of maximum transmit power. Based on this mechanism, the analytical MSE expression of AirComp is derived by Campbell’s theorem and validated by Monte Carlo simulations. Based on the analytical MSE result, the effect of the AP access radius on the MSE are analysed numerically. The results show that an optimal AP access radius exists for AirComp, which can reduce a maximum of $12.7\%$ of the MSE. Consequently, the AP access radius should be carefully chosen in the CITN, which can significantly improve AirComp performance in the ultra-dense deployment of IoT devices.  

\section{SYSTEM MODEL}\label{net_model}
\subsection{Network Model}
\begin{figure}
\centering
\includegraphics[width=6cm]{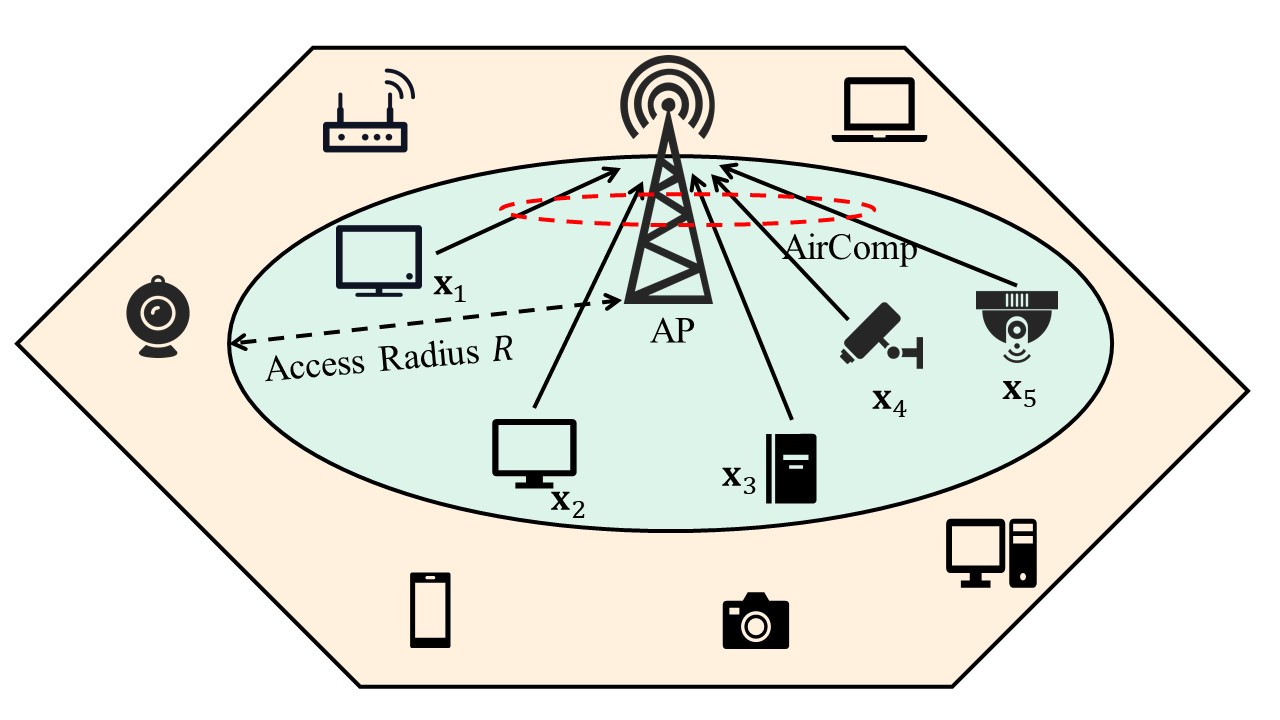}
\caption{\label{fig:fig1}A single-cell IoT network with AirComp.}
\end{figure}

We consider a single-cell CITN consisting of a single-antenna AP and multiple IoT devices\footnote{Since we focus on the AirComp performance in the CITN, for simplicity, it is assumed that the AP is equipped with a single antenna, which can still provide design insights on the CITN with AirComp \cite{10050442,9158563,9624970}.}, as illustrated in Fig. \ref{fig:fig1}. Without loss of generality, we place the AP at the origin, and all IoT devices are located following the PPP with density being $\lambda$ \cite{probability}. Meanwhile, due to limited transmit power of IoT devices, we assume that the IoT devices within an access radius $R$ of the AP are able to perform AirComp. In this case, if there are $K$ IoT devices within the AP access radius, these IoT devices are represented by $\{\boldsymbol{{\rm{x}}}_{k}\}$, $k\in(1,2,\cdots,K)$, where $k$ denotes the $k$-th IoT device within the AP access radius.
 
For the $k$-th IoT device $\boldsymbol{{\rm{x}}}_{k}$, the Euclidean distance between it and the AP can be denoted by $d_{k}=\|\boldsymbol{{\rm{x}}}_{k}\|$. We assume that the path-loss follows a log-distance model, which can be expressed by $l(d_k)=d_{k}^{-\alpha}$, where $\alpha$ denotes the path-loss exponent. 

All IoT devices can transmit their measured data to the AP with an uplink transmit power $P$. The uplink transmit power is assumed to adopt a fractional channel inversion power control mechanism given by $P=pl(u)^{-\epsilon}$ \cite{9158563}, where $p$ denotes the fixed transmit power, $u$ is the distance between the IoT device and the AP, and $\epsilon$ denotes the power control factor. In practice, all IoT devices are constrained by a maximum transmit power $P_{\max}$, i.e., $P\leq P_{\max}$. Consequently, the uplink transmit power $P_k$ of the $k$-th IoT device $\boldsymbol{{\rm{x}}}_{k}$ can be expressed by 
\begin{equation}\label{p_k}
    P_{k}= \min (P_{\max}, pd_k^{\alpha\epsilon}).
\end{equation}

For the small-scale fading, due to the relatively short communication distances in the ultra-densely deployed CITN, Rician fading is assumed. As a result, the power attenuation $h_k$ of small-scale fading of the link from the $k$-th IoT device $\boldsymbol{{\rm{x}}}_{k}$ to the AP can be modeled by \cite{mimo} 
\begin{equation}\label{rician}
    h_{k} = \sqrt{\frac{B}{B+1}}h^{\rm{L}}_{k}+\sqrt{\frac{1}{B+1}}h^{\rm{NL}}_{k},
\end{equation}
where $h^{\rm{L}}_{k}$ represents the line-of-sight (LoS) of small-scale fading, which is modeled as $h^{\rm{L}}_{k}=e^{-j\frac{2\pi d_{k}}{W}}$ \cite{9624970}, where $W$ denotes the wavelength. $h^{\rm{NL}}_{k}$ denotes the non-line-of-sight (NLoS) of small-scale fading, which follows a Rayleigh fading channel. Additionally, $B$ represents the Rician factor, which can be calculated by $\frac{c^2}{2\sigma^2}$, where $c^2$ is the power of the LoS and $\sigma^2$ is the power of the NLoS. Consequently, $|h_{k}|=v$ follows a Rician distribution with a probability density function (PDF) that can be expressed by \cite{mimo}
\begin{equation}\label{rician_pdf}
 f_{h_{k}}(v)=\frac{v}{\sigma^{2}}\exp\left(-\frac{v^{2}+c^2}{2\sigma^{2}}\right)I_{0}\left(\frac{vc}{\sigma^2}\right),  v\geq 0, 
\end{equation}
where $c=\sqrt{\frac{B}{B+1}}$ and $\sigma=\sqrt{\frac{1}{2(B+1)}}$. $I_{0}(\cdot)$ denotes the modified zeroth-order Bessel function of the first kind \cite{bessel}.

\subsection{Over-the-Air Computation}
To achieve ultra-fast WDA of distributed data, AirComp is applied in the ultra-densely deployed CITN. We assume that all IoT devices within the access radius of the AP are perfectly synchronized via exploiting the AirShare technique \cite{7218555}. The function for receiving aggregated data at the AP can be expressed as $\widetilde{f}(k)=\sum_{k=1}^{K}z_k$, where $z_k$ represents the data measured by the $k$-th IoT device $\boldsymbol{{\rm{x}}}_{k}$. To eliminate amplitude differences in the measured data between different IoT devices and facilitate the analysis of the measured data, $z_k$ is normalized as $s_k=g(z_k)$, where $g(\cdot)$ denotes the normalization function to obtain signal $\{s_k\}$ with zero mean and unit variance \cite{9158563}. As \cite{9095231}, we assume the random variables $\{s_k\}$ to be independently and identically distributed (i.i.d) for different IoT devices. Consequently, the function for receiving the aggregated data at the AP can be further expressed as $f(k)=\sum_{k=1}^{K}s_k$. 

The received signal at the AP, denoted by $y$, can be expressed by 
\begin{equation}\label{eq4}
    y=\sum_{k=1}^{K}d_{k}^{-\frac{\alpha}{2}}h_{k}b_{k}s_{k}+n,
\end{equation}
where $n$ is the additive white Gaussian noise (AWGN) with zero mean and variance of $\omega^{2}$, $b_k$ represents the transmit coefficient of the $k$-th IoT device $\boldsymbol{{\rm{x}}}_{k}$, which can be calculated by $\frac{\sqrt{P_k}h_{k}^{\dag}}{|h_k|}$. In this parameter, $h_{k}^{\dag}$ denotes the conjugate function of $h_k$. Consequently, the received signal at the AP can be expressed by 
\begin{equation}
    y=\sum_{k=1}^{K}d_{k}^{-\frac{\alpha}{2}}\sqrt{P_{k}}|h_{k}|s_{k}+n.
\end{equation}
Accordingly, the estimated function recovered by the AP can be expressed by
\begin{equation}
    \widehat{f}(k)=\frac{y}{\sqrt{\eta}}=\frac{\sum_{k=1}^{K}d_{k}^{-\frac{\alpha}{2}}\sqrt{P_{k}}|h_{k}|s_{k}}{\sqrt{\eta}}+\frac{n}{\sqrt{\eta}},
\end{equation}
where $\eta$ denotes the denoising factor at the AP. The AirComp performance can be evaluated by the MSE between the function $f(k)$ and the corresponding estimated function $\widehat{f}(k)$, which can be expressed by \cite{zhou2021machine}
\begin{equation}\label{MSE}
   {\rm{MSE}} = \mathbb{E}\left[\frac{1}{K}\left(\widehat{f}(k)-f(k)\right)^{2} \right],
\end{equation}
where $\mathop{\mathbb{E}}[\cdot]$ represents the expectation with respect to the distribution of the random variables $K$, $\{s_k\}$, $\{d_k\}$, $\{h_k\}$, and $n$.

\section{PERFORMANCE ANALYSIS}
In this section, we first derive the analytical expression of MSE in the CITN. From \eqref{MSE}, the MSE of AirComp at the AP can be calculated by  
\begin{equation}\label{MSE_A}
    \begin{aligned}
  & {\rm{MSE}} =\mathbb{E}\left[\frac{1}{K} \left(\frac{y}{\sqrt{\eta}} - \sum_{k\in \mathcal{K}}s_k\right)^{2} \right]\\
    & =\mathbb{E}\left[\frac{1}{K}\left(\sum_{k=1}^{K}s_{k} \left(\frac{d_{k}^{-\frac{\alpha}{2}} \sqrt{P_{k}} |h_{k}|}{\sqrt{\eta}} -1 \right)+\frac{n}{\sqrt{\eta}} \right)^{2}\right]\\
    & \overset{(a)}{=} \mathbb{E}\left[\frac{1}{K}\left(\sum_{k=1} ^{K}\left(\frac{d_{k}^{-\frac{\alpha}{2}} \sqrt{P_{k}} |h_{k}|}{\sqrt{\eta}} - 1 \right)^{2} + \frac{\omega^{2}}{\eta}\right) \right],
    \end{aligned}
\end{equation}
where step $(a)$ is obtained by taking the expectation of the random variables $\{s_k\}$ and $n$. 

Considering the limited transmit power of IoT devices as described in \eqref{p_k}, the value of the transmit power $P_{k}$ in \eqref{MSE_A} is affected by two factors: the link distance $d_k$ between the $k$-th IoT device $\boldsymbol{{\rm{x}}}_{k}$ and the AP, and the power attenuation $h_k$ of small-scale fading of the link. Specifically, if the magnitude of $|h_k|$ exceeds the threshold $\sqrt{\frac{\eta}{P_{\max}}}d_{k}^{\frac{\alpha\epsilon}{2}}$, the transmit power of the IoT device $\boldsymbol{{\rm{x}}}_{k}$ does not exceed the maximum transmit power $P_{\max}$. In this case, the transmit power of the IoT device $\boldsymbol{{\rm{x}}}_{k}$ can be expressed as $pd_{k}^{\alpha\epsilon}$, where $p=\frac{\eta}{|h_{k}|^{2}}$ \cite{9158563}. Otherwise, the maximum transmit power will be employed, as follows. 
\begin{equation}\label{power}
    \begin{split}
     P_k= \left\{
    \begin{array}{ll}
    P_{\max}, & 0< |h_{k}| \leq \sqrt{\frac{\eta}{P_{\max}}}d_{k}^{\frac{\alpha\epsilon}{2}},\\
    p d_{k}^{\alpha\epsilon}, & |h_{k}| > \sqrt{\frac{\eta}{P_{\max}}}d_{k}^{\frac{\alpha\epsilon}{2}}.
    \end{array}
    \right.    
    \end{split}
\end{equation}

Equipped with \eqref{power}, the analytical result of MSE of AirComp in the CITN is given in Theorem \ref{Theorem} as follows.
\begin{MyTheo}\label{Theorem}
In the CITN with IoT devices deployed following PPP under the limited transmit power of IoT devices, the MSE of AirComp can be given by
\begin{equation}
\begin{aligned}
    &{\rm{MSE}}=\mathcal{K}(R)  \left\{ 2\pi\lambda \left[\int_{1}^{R} \int_{0}^{D(r)} \left(\frac{r^{1-\alpha} P_{\max} v^{2}}{\eta} \right.\right.\right.\\
        &\left. - \frac{2r^{1-\frac{\alpha}{2}}\sqrt{P_{\max}} v}{\sqrt{\eta}}\right)f_{h_{k}}(v){\rm{d}}v {\rm{d}}r + \frac{1}{2}(R^2-1)\\
        &\left.\left. +\int_{1}^{R}\left(r^{\alpha\epsilon-\alpha} - 2r^{\frac{\alpha\epsilon-\alpha}{2}} + 2\right)r Q_{1}\left(\frac{c}{\sigma},\frac{D(r)}{\sigma} \right) {\rm{d}}r \right] + \frac{\omega^{2}}{\eta}\right\},
        \end{aligned}
\end{equation}
where the function $\mathcal{K}(R)$ is defined as
\begin{equation}
    \mathcal{K}(R) \triangleq  \exp\left(-\lambda\pi R^{2}\right)\left[{\rm{ei}}\left(\lambda\pi R^{2}\right) - {\rm{In}}\left(\lambda\pi R^{2}\right)\right].
\end{equation}
In these parameters, ${\rm{ei}}\left(\lambda\pi R^{2}\right)=\int\frac{e^{\lambda\pi R^{2}}}{\lambda\pi R^{2}}{\rm{d}}\left(\lambda\pi R^{2}\right)$. In addition, $D(r)=\sqrt{\frac{\eta}{P_{\max}}}r^{\frac{\alpha\epsilon}{2}}$, the function $f_{h_{k}}(v)$ is given in \eqref{rician_pdf}, and $Q_{1}(a,b)$ denotes Marcum Q-function, $Q_{1}(a,b)=\int_{b}^{\infty}u\exp\left(-\frac{u^2+a^2}{2}\right)I_{0}(au){\rm{d}}u$
\end{MyTheo}
\begin{proof}
    See Appendix \ref{proof_mse}.
\end{proof}

In this work, we aim to minimize the MSE of AirComp in the CITN. It can be observed that the denoising factor $\eta$ may have significant effects on the MSE. As a result, with the analytical expression of MSE given in Theorem \ref{Theorem}, we can obtain the minimal MSE by solving the following problem as follows.
\begin{equation}\label{eq_min_mes}
\begin{split}
&\arg \min_{\{\eta\}} \,\,{\rm{MSE}} \\
& \begin{array}{l@{\quad}l@{}l@{\quad}r}
{\rm{s.t.}} & 0<\eta \leq \widehat\eta(R),\\
\end{array}
\end{split}
\end{equation}
where the feasible range of $\eta$ is given in Corollary \ref{Corollary}. 
\begin{MyCoro}\label{Corollary}
For the MSE of AirComp in the CITN, the maximum denoising factor $\eta$ can be expressed by
\begin{equation}
     \widehat\eta(R)= \max\left(\frac{P_{\max}(3R^2)^{\alpha\epsilon}}{(2R^3-2)^{\alpha\epsilon}},  \frac{\left(2\pi \lambda P_{\max}\frac{R^{2-\alpha}-1}{2-\alpha}+\omega^2\right)^2}{\left(2\pi \lambda \sqrt{P_{\max}} \frac{R^{2-\frac{\alpha}{2}}-1}{2-\frac{\alpha}{2}}\sqrt{\frac{\pi}{2}}\sigma\right)^2}\right).
\end{equation}
\end{MyCoro}
\begin{proof}
    See Appendix \ref{eta_proof}.
\end{proof}
As the denoising factor $\eta$ has a limited range, the minimal MSE can be obtained by the bisection search over the range of $\eta$.

\section{SIMULATION RESULTS}
In this section, the analytical result of MSE that is given in Theorem \ref{Theorem} will be firstly validated by Monte Carlo simulations. In each iteration of the Monte Carlo simulation, the locations of IoT devices are deployed following the PPP model described in Section \ref{net_model} in a $100\times100$ m$^2$ square area. Then the MSE of AP located at the origin can be calculated following \eqref{MSE_A} in each iteration. We run $10,000$ iterations and obtain $10,000$ MSEs of the AP. Consequently, the MSE can be obtained by calculating the arithmetic mean of in \eqref{MSE_A}. After the validation, the MSE will be analyzed in terms of the AP access radius numerically.
 
\begin{figure}[!t]
\centering
\includegraphics[width=0.35\textwidth]{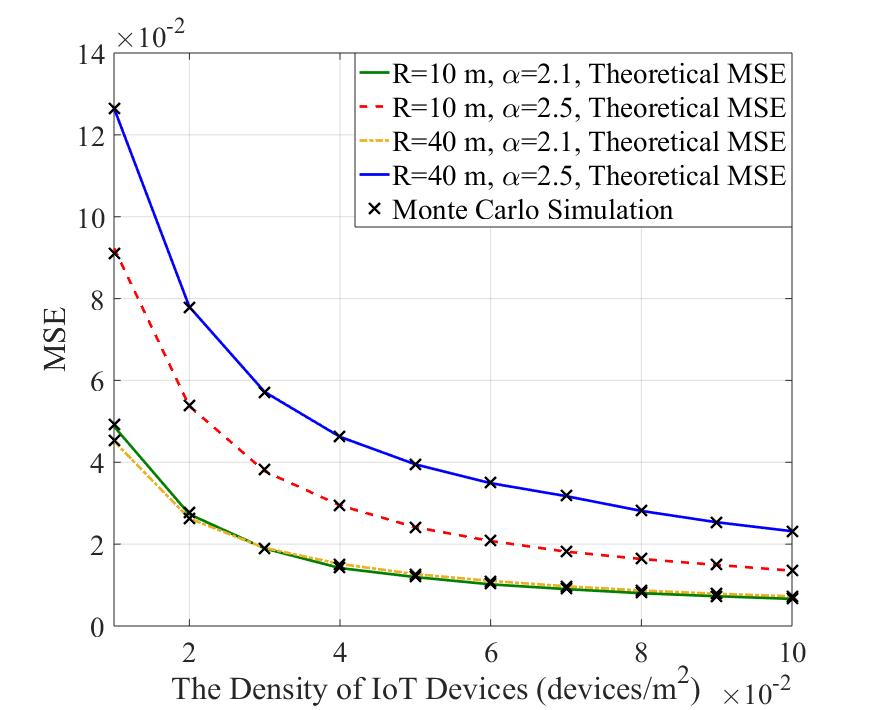}
\caption{\label{fig:fig2} The validation for the MSE versus the density of IoT devices, where $L=0.3$ m, $\omega^2=1$, SNR$=\frac{P_{\max}}{\omega^2}=30$ dB, $B=15$, and $\epsilon=1$.
}
\end{figure}

In Fig. \ref{fig:fig2}, the theoretical and simulation results of MSE in the CITN are plotted versus the density of IoT devices $\lambda$. The theoretical MSEs are matched with the simulation results, verifying the correctness of our derived results. The MSE is inversely proportional to the density of IoT devices $\lambda$. With the increasing of the density of IoT devices, the AP can aggregate more data, which reduces the MSE. Additionally, it is observed that when the path-loss exponent is $2.1$, the MSE with a lower density of IoT devices is smaller at $R=40$ m. Conversely, the MSE with a higher density of IoT devices is smaller at $R=10$ m. It indicates that the AP access radius has a significant impact on the MSE. 
 
\begin{figure}[!t]
\centering
\includegraphics[width=0.35\textwidth]{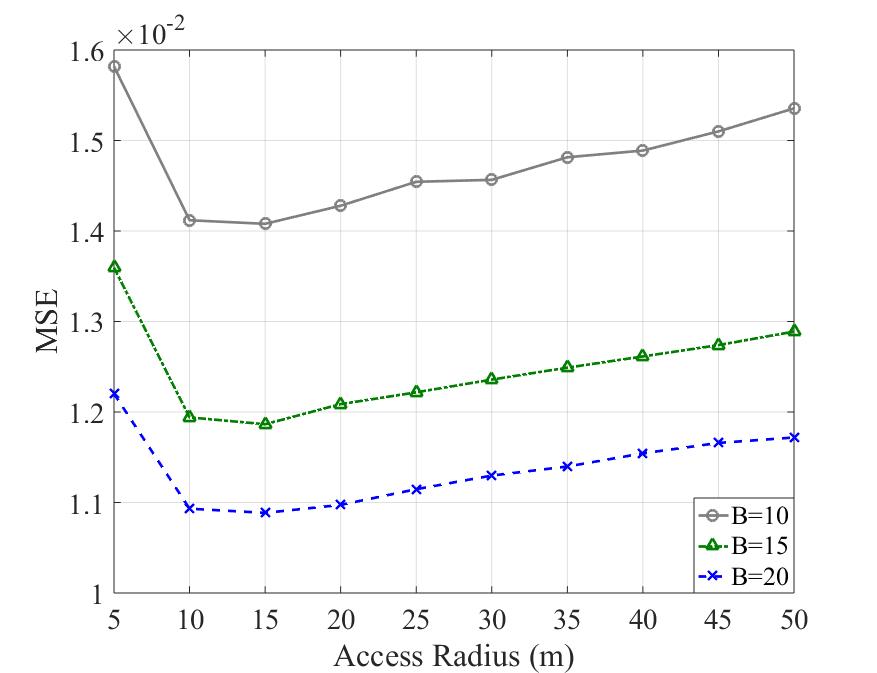}
\caption{\label{fig:fig3} The theoretical MSE versus the AP access radius $R$, where $\lambda=5\times10^{-2}\rm{devices/m^2}$, $\alpha=2.1$ $L=0.3$ m, $\omega^2=1$, SNR$=\frac{P_{\max}}{\omega^2}=30$ dB, and $\epsilon=1$  
}
\end{figure}

In Fig. \ref{fig:fig3}, the theoretical MSEs under $B$ being $10$, $15$, and $20$ are plotted versus the AP access radius $R$. It is observed that an optimal AP access radius exists for minimizing the MSE. Under our simulation environment, with the increasing of AP access radius, a sharp decrease in terms of the MSE occurs when the AP access radius is $15$ m and then the MSE increases, thus the optimal AP access radius occurs near $R=15$ m with the $B$ being $10$, $15$, and $20$. The result indicates that the AP access radius has a significant impact on the MSE performance. The lower AP access radius results in worst-case MSE due to the limitation of the number of IoT devices. Similarly, The higher AP access radius leads to worst-case MSE due to the limitation of transmit power of IoT devices. Additionally, as compared with $R=5$ m, we can observe that the MSE with the optimal AP access radius can be reduced by maximally $12.7\%$ under the Rician factor $B$ being $15$. This result indicates that by choosing an appropriate AP access radius, the AirComp performance can be reduce in the ultra-densely deployed CITN.
 
\section{CONCLUSION}
\label{section_conclusion}
In this work, we have obtained the analytical results for the MSE of AirCommp in the CITN. Based on these numerical results, we have numerically analyzed the effects of the AP access radius and the density of IoT devices on the MSE. Specifically, the minimal MSE can be obtained via the bisection method by given the feasible range of the denoising factor. It has been shown that the MSE decreases with the increasing density of IoT devices. Additionally, we can find an optimal AP access radius that significantly reduces the MSE in the densely deployed CITN. Under our simulation environment, the MSE of AirComp can be reduced by $12.7\%$ at the optimal AP access radius, which is located at $15$ m. Consequently, it is crucial to carefully choose the AP access radius in the CITN to enhance the AirComp performance. For future work, we can investigate large-scale AirComp deployments over multiple APs in the CITN, which need to consider the trade-off between inter-cell interference and the MSE. 

\appendices
\section{Proof of Theorem \ref{Theorem}}\label{proof_mse}
By incorporating \eqref{power} into \eqref{MSE_A}, the MSE of AirComp can be calculated by 
\begin{equation}\label{MSE1}
    \begin{aligned}
        &{\rm{MSE}}=\mathbb{E}\left[\frac{1}{K}\right] \mathop{\mathbb{E}} \limits_{\{d_k\}, \{h_k\}}\left[\sum_{k=1}^{K} \left(\frac{d_{k}^{-\frac{\alpha}{2}} \sqrt{P_{k}} |h_{k}|}{\sqrt{\eta}} - 1 \right)^{2} + \frac{\omega^{2}}{\eta} \right] \\
        & \overset{(b)}{=} \mathbb{E}\left[\frac{1}{K}\right] \mathop{\mathbb{E}} \limits_{\{h_k\}} \left[ 2\pi\lambda \int_{1}^{R} \left(\frac{r^{-\frac{\alpha}{2}}\sqrt{P_{k}} |h_k|}{\sqrt{\eta}} -1\right)^{2}r {\rm{d}}r + \frac{\omega^{2}}{\eta}\right]\\
        & = \mathbb{E}\left[\frac{1}{K}\right]  \left\{ 2\pi\lambda \left[  \int_{1}^{R} \int_{0}^{D(r)} \left(\frac{r^{-\frac{\alpha}{2}}\sqrt{P_{\max}} v}{\sqrt{\eta}} -1\right)^{2}r f_{h_{k}}(v){\rm{d}}v {\rm{d}}r\right.\right.\\
        & \left.\left. +  \int_{1}^{R} \int_{D(r)}^{\infty} \left(r^{\frac{\alpha\epsilon-\alpha}{2}} - 1\right)^{2} r f_{h_{k}}(v){\rm{d}}v {\rm{d}}r \right] + \frac{\omega^{2}}{\eta}\right\}\\
        & = \mathbb{E}\left[\frac{1}{K}\right]  \left\{ 2\pi\lambda \left[\int_{1}^{R} \int_{0}^{D(r)} \left(\frac{r^{1-\alpha} P_{\max} v^{2}}{\eta} \right.\right.\right.\\
        &\left. - \frac{2r^{1-\frac{\alpha}{2}}\sqrt{P_{\max}} v}{\sqrt{\eta}}\right)f_{h_{k}}(v){\rm{d}}v {\rm{d}}r + \frac{1}{2}(R^2-1)\\
        &\left.\left. +\int_{1}^{R}\left(r^{\alpha\epsilon-\alpha} - 2r^{\frac{\alpha\epsilon-\alpha}{2}} + 2\right)r Q_{1}\left(\frac{c}{\sigma},\frac{D(r)}{\sigma} \right) {\rm{d}}r \right] + \frac{\omega^{2}}{\eta}\right\},
    \end{aligned}
\end{equation}
where step $(b)$ is obtained following the Campbell's theorem for sums \cite{probability}, $D(r)=\sqrt{\frac{\eta}{P_{\max}}}r^{\frac{\alpha\epsilon}{2}}$, $f_{h_{k}}(v)$ is the PDF of the Rician distribution as given in \eqref{rician_pdf}, and $Q_{1}(a,b)$ denotes the Marcum Q-function defined as $Q_{1}(a,b)=\int_{b}^{\infty}u\exp\left(-\frac{u^2+a^2}{2}\right)I_{0}(au){\rm{d}}u$. It is worth noting that within the AP access radius of $1$ m, there is no path-loss effect. Therefore, the integration range of $d_{k}$ is from $1$ to $R$.

Additionally, the PDF of the number of IoT devices located within the range of AP access radius can be given by \cite{probability}
\begin{equation}
    \mathbb{P}[K=m]=\frac{\left(\lambda\pi R^{2}\right)^{m}}{m!}\exp\left(-\lambda\pi R^{2}\right).
\end{equation}
As a result, the expectation $\mathbb{E}\left[\frac{1}{K}\right]$ in equation \eqref{MSE1} can be derived by
\begin{equation}\label{K}
    \begin{aligned}
        \mathbb{E}\left[\frac{1}{K}\right]&=\sum_{m=1}^{\infty}\frac{P(K=m)}{m} = \exp\left(-\lambda\pi R^{2}\right)\sum_{m=1}^{\infty}\frac{\left(\lambda\pi R^{2}\right)^{m}}{m m!}\\
        & = \exp\left(-\lambda\pi R^{2}\right) \int \frac{1}{\lambda\pi R^{2}}\sum_{m=1}^{\infty}\frac{\left(\lambda\pi R^{2}\right)^{m}}{m!} {\rm{d}}\left(\lambda\pi R^{2}\right) \\
        & = \exp\left(-\lambda\pi R^{2}\right) \int \frac{1}{\lambda\pi R^{2}} (e^{\lambda\pi R^{2}}-1) {\rm{d}}\left(\lambda\pi R^{2}\right)\\
        & = \exp\left(-\lambda\pi R^{2}\right)\left[{\rm{ei}}\left(\lambda\pi R^{2}\right) - {\rm{In}}\left(\lambda\pi R^{2}\right)\right] \\
        & \triangleq \mathcal{K}(R),
    \end{aligned}
\end{equation}
where ${\rm{ei}}\left(\lambda\pi R^{2}\right)=\int\frac{e^{\lambda\pi R^{2}}}{\lambda\pi R^{2}}{\rm{d}}\left(\lambda\pi R^{2}\right)$. 

By incorporating \eqref{K} into \eqref{MSE1}, the analytical result of MSE in Theorem \ref{Theorem} can be obtained.

\section{Proof of Corollary \ref{Corollary}}\label{eta_proof}
According to \eqref{power}, the feasible range of $\eta$ can be obtained by
\begin{equation}\label{p_eta}
    \begin{split}
     P_k^{\ast}= \left\{
    \begin{array}{ll}
    P_{\max}, & P_{\max}|h_k|^2 d_{k}^{\alpha\epsilon}<\eta\leq \mathcal{H}_{\eta}(d_k, h_k),\\
    \frac{\eta}{|h_{k}|^{2}}d_{k}^{\alpha\epsilon}, & 0<\eta\leq P_{\max}|h_k|^2 d_{k}^{\alpha\epsilon},
    \end{array}
    \right.    
    \end{split}
\end{equation}
where the function $\mathcal{H}_{\eta}(d_k, h_k)=\max \left(P_{\max}|h_k|^2 d_{k}^{\alpha\epsilon} ,\eta^{\ast}(d_k, h_k)\right)$. In these parameters, the function $\eta^{\ast}(d_k, h_k)$ is defined as the extreme point of function $\mathcal{G}(\eta)$, where the function $\mathcal{G}(\eta)$ from \eqref{MSE_A} is defined as
\begin{equation}
    \mathcal{G}(\eta)=\frac{1}{K}\left(\sum_{k=1} ^{K}\left(\frac{d_{k}^{-\frac{\alpha}{2}} \sqrt{P_{k}} |h_{k}|}{\sqrt{\eta}} - 1 \right)^{2} + \frac{\omega^{2}}{\eta}\right).
\end{equation}
By computing the partial differentiation $\frac{ \partial \mathcal{G}(\eta) }{ \partial \eta }=0$, we can obtain the extreme point $\eta^{\ast}(d_k,h_k)$ of $\mathcal{G}(\eta)$ , which can be expressed by
\begin{equation}
    \eta^{\ast}(d_k, h_k) = \left(\frac{\sum_{k=1}^{K}d_k^{-\alpha}P_k|h_k|^2+\sigma_{n}^2}{\sum_{k=1}^{K}d_k^{-\frac{\alpha}{2}}\sqrt{P_k}|h_k|}\right)^2.
\end{equation}
It can be observed that $\frac{ \partial \mathcal{G}(\eta) }{ \partial \eta }<0$ for $\eta<\eta^{\ast}(d_k, h_k)$ and $\frac{ \partial \mathcal{G}(\eta) }{ \partial \eta }>0$ for $\eta>\eta^{\ast}(d_k, h_k)$. Therefore, the function $\mathcal{G}(\eta)$ is monotonically decreasing in $\eta\in (0,\eta^{\ast}(d_k, h_k)]$ and monotonically increasing in $\eta\in [\eta^{\ast}(d_k, h_k), \infty]$. As a result, $\mathcal{G}(\eta^{\ast}(d_k, h_k))$ is the minimum of the function $\mathcal{G}(\eta)$.

Consequently, according to \eqref{p_eta}, the whole feasible range of $\eta$ can be expressed by $\eta\in (0,\mathcal{H}_{\eta}(d_k, h_k)]$. With the independent variables $d_k$ and $h_k$, the expectation $\mathop{\mathbb{E}} \limits_{\{d_k\}, \{h_k\}}[\mathcal{H}_{\eta}(d_k, h_k)]$ need to be derived. First, the expectation of $P_{\max}|h_k|^2 d_{k}^{\alpha\epsilon}$ can be calculated by
\begin{equation}\label{eta1}
    \begin{aligned}     
    &\mathop{\mathbb{E}} \limits_{\{d_k\}, \{h_k\}}\left[P_{\max}|h_k|^2 d_{k}^{\alpha\epsilon}\right]=P_{\max} \left(\int_{1}^{R}r f_{d_k}(r){\rm{d}}r\right)^{\alpha\epsilon}\\
    &= P_{\max}\left[\frac{2(R^3-1)}{3R^2}\right]^{\alpha\epsilon},
    \end{aligned}
\end{equation}
where the PDF of the link distance $d_k$ between the $k$-th IoT device $\boldsymbol{{\rm{x}}}_{k}$ and the AP can be given by $f_{d_k}(r)=\frac{2r}{R^2}$ \cite{9951137}, and $\mathbb{E}[|h_k|^2]=1$. Additionally, when $P_k=P_{\max}$, the expectation of $\eta^{\ast}(d_k, h_k)$ can be calculated by
\begin{equation}\label{eta2}
    \begin{aligned}     
    &\mathop{\mathbb{E}} \limits_{\{d_k\}, \{h_k\}}[\eta^{\ast}(d_k, h_k)]\\
    & =\mathop{\mathbb{E}} \limits_{\{d_k\}, \{h_k\}}\left[\left(\frac{\sum_{k=1}^{K}d_k^{-\alpha}P_{\max}|h_k|^2+\omega^2}{\sum_{k=1}^{K}d_k^{-\frac{\alpha}{2}}\sqrt{P_{\max}}|h_k|}\right)^2\right]\\ 
    & \overset{(c)}{=}  \left(\frac{2\pi\lambda P_{\max}\int_{1}^{R}\int_{0}^{\infty} r^{1-\alpha}v^2 f_{h_k}(v){\rm{d}}v{\rm{d}}r +\omega^2}{2\pi\lambda \sqrt{P_{\max}}\int_{1}^{R}\int_{0}^{\infty} r^{1-\frac{\alpha}{2}}v f_{h_k}(v){\rm{d}}v {\rm{d}}r}\right)^2\\
    & =  \left(\frac{2\pi \lambda P_{\max}\frac{R^{2-\alpha}-1}{2-\alpha}+\omega^2}{2\pi \lambda \sqrt{P_{\max}} \frac{R^{2-\frac{\alpha}{2}}-1}{2-\frac{\alpha}{2}}\sqrt{\frac{\pi}{2}}\sigma}\right)^2,
    \end{aligned}
\end{equation}
where step $(c)$ is obtained following the Campbell's theorem for sums \cite{probability}.

By combining \eqref{eta1} and \eqref{eta2}, the feasible range of $\eta$ can be achieved in Corollary \ref{Corollary}.

%\bibliographystyle{IEEEtran}
%\bibliography{FILE}

% Generated by IEEEtran.bst, version: 1.14 (2015/08/26)
\begin{thebibliography}{10}
\providecommand{\url}[1]{#1}
\csname url@samestyle\endcsname
\providecommand{\newblock}{\relax}
\providecommand{\bibinfo}[2]{#2}
\providecommand{\BIBentrySTDinterwordspacing}{\spaceskip=0pt\relax}
\providecommand{\BIBentryALTinterwordstretchfactor}{4}
\providecommand{\BIBentryALTinterwordspacing}{\spaceskip=\fontdimen2\font plus
\BIBentryALTinterwordstretchfactor\fontdimen3\font minus
  \fontdimen4\font\relax}
\providecommand{\BIBforeignlanguage}[2]{{%
\expandafter\ifx\csname l@#1\endcsname\relax
\typeout{** WARNING: IEEEtran.bst: No hyphenation pattern has been}%
\typeout{** loaded for the language `#1'. Using the pattern for}%
\typeout{** the default language instead.}%
\else
\language=\csname l@#1\endcsname
\fi
#2}}
\providecommand{\BIBdecl}{\relax}
\BIBdecl

\bibitem{9481289}
L.~Su and V.~K.~N. Lau, ``Data and channel-adaptive sensor scheduling for
  federated edge learning via over-the-air gradient aggregation,'' \emph{IEEE
  Internet of Things Journal}, vol.~9, no.~3, pp. 1640--1654, 2022.

\bibitem{10038617}
J.~Du, B.~Jiang, C.~Jiang, Y.~Shi, and Z.~Han, ``Gradient and channel aware
  dynamic scheduling for over-the-air computation in federated edge learning
  systems,'' \emph{IEEE Journal on Selected Areas in Communications}, vol.~41,
  no.~4, pp. 1035--1050, 2023.

\bibitem{9095231}
W.~Liu, X.~Zang, Y.~Li, and B.~Vucetic, ``Over-the-air computation systems:
  Optimization, analysis and scaling laws,'' \emph{IEEE Transactions on
  Wireless Communications}, vol.~19, no.~8, pp. 5488--5502, 2020.

\bibitem{6557530}
M.~Goldenbaum, H.~Boche, and S.~Stańczak, ``Harnessing interference for analog
  function computation in wireless sensor networks,'' \emph{IEEE Transactions
  on Signal Processing}, vol.~61, no.~20, pp. 4893--4906, 2013.

\bibitem{10092857}
A.~Şahin and R.~Yang, ``A survey on over-the-air computation,'' \emph{IEEE
  Communications Surveys and Tutorials}, pp. 1--1, 2023.

\bibitem{4305404}
B.~Nazer and M.~Gastpar, ``Computation over multiple-access channels,''
  \emph{IEEE Transactions on Information Theory}, vol.~53, no.~10, pp.
  3498--3516, 2007.

\bibitem{4655452}
M.~Gastpar, ``Uncoded transmission is exactly optimal for a simple gaussian
  “sensor” network,'' \emph{IEEE Transactions on Information Theory},
  vol.~54, no.~11, pp. 5247--5251, 2008.

\bibitem{8708985}
X.~Li, G.~Zhu, Y.~Gong, and K.~Huang, ``Wirelessly powered data aggregation for
  iot via over-the-air function computation: Beamforming and power control,''
  \emph{IEEE Transactions on Wireless Communications}, vol.~18, no.~7, pp.
  3437--3452, 2019.

\bibitem{9097897}
X.~Zang, W.~Liu, Y.~Li, and B.~Vucetic, ``Over-the-air computation systems:
  Optimal design with sum-power constraint,'' \emph{IEEE Wireless
  Communications Letters}, vol.~9, no.~9, pp. 1524--1528, 2020.

\bibitem{9158563}
X.~Cao, G.~Zhu, J.~Xu, and K.~Huang, ``Optimized power control for over-the-air
  computation in fading channels,'' \emph{IEEE Transactions on Wireless
  Communications}, vol.~19, no.~11, pp. 7498--7513, 2020.

\bibitem{9298474}
------, ``Cooperative interference management for over-the-air computation
  networks,'' \emph{IEEE Transactions on Wireless Communications}, vol.~20,
  no.~4, pp. 2634--2651, 2021.

\bibitem{9624970}
H.~Jung and S.-W. Ko, ``Performance analysis of uav-enabled over-the-air
  computation under imperfect channel estimation,'' \emph{IEEE Wireless
  Communications Letters}, vol.~11, no.~3, pp. 438--442, 2022.

\bibitem{10050442}
S.~Tang, P.~Popovski, C.~Zhang, and S.~Obana, ``Multi-slot over-the-air
  computation in fading channels,'' \emph{IEEE Transactions on Wireless
  Communications}, pp. 1--1, 2023.

\bibitem{probability}
M.~Haenggi, \emph{Stochastic Geometry for Wireless Networks}.\hskip 1em plus
  0.5em minus 0.4em\relax Cambridge: Cambridge University Press, 2012.

\bibitem{mimo}
W.~Y.~Y. Yong Soo~Cho, Jaekwon~Kim and C.~G. Kang, \emph{MIMO-OFDM Wireless
  Communications with MATLAB}.\hskip 1em plus 0.5em minus 0.4em\relax John
  Wiley and Sons, 2010.

\bibitem{bessel}
F.~Mainardi, \emph{Fractional Calculus and Waves in Linear
  Viscoelasticity}.\hskip 1em plus 0.5em minus 0.4em\relax London: Imperial
  College Press, 2010.

\bibitem{7218555}
O.~Abari, H.~Rahul, D.~Katabi, and M.~Pant, ``Airshare: Distributed coherent
  transmission made seamless,'' in \emph{2015 IEEE Conference on Computer
  Communications (INFOCOM)}, 2015, pp. 1742--1750.

\bibitem{zhou2021machine}
Z.-H. Zhou, \emph{Machine learning}.\hskip 1em plus 0.5em minus 0.4em\relax
  Springer Nature, 2021.

\bibitem{9951137}
H.~Hu, Y.~Dong, Y.~Jiang, Q.~Chen, and J.~Zhang, ``On the age of information
  and energy efficiency in cellular iot networks with data compression,''
  \emph{IEEE Internet of Things Journal}, vol.~10, no.~6, pp. 5226--5239, 2023.

\end{thebibliography}

\end{document}